\documentclass[aps,prd,superscriptaddress,12pt,showpacs,notitlepage]{revtex4}
\usepackage{amsfonts}
	\usepackage{graphicx}
	\usepackage{amsmath,amssymb,mathrsfs}
\usepackage{epsf}
\usepackage{epsfig}

\usepackage{amsmath}

\usepackage{amssymb}

\newcommand{\R}{{\mathbb{R}}} 
\newcommand{\ra}{\rightarrow} 
\newcommand{\ip}[1]{\langle #1 \rangle }

\def \vol{\mathrm{vol}}
\def \d{\mathrm{d}}
\def \grad{\mathrm{grad}}

\newtheorem{prop}{Proposition}
\newtheorem{lemma}{Lemma}
\newenvironment{proof}{\noindent {\it Proof:\,
}}{\hfill$\Box$\vspace*{0.4cm}}

\newcommand{\be}{\begin{equation}}
\newcommand{\ee}{\end{equation}}
\newcommand{\bea}{\begin{eqnarray}}
\newcommand{\eea}{\end{eqnarray}}

\newcommand{\bb}{\bibitem}
 
\newcommand{\eqn}{\begin{eqnarray}}
\newcommand{\eqnx}{\end{eqnarray}}

\begin{document}
\title{Thermodynamics of the BPS Skyrme model}

\author{C. Adam}
\affiliation{Departamento de F\'isica de Part\'iculas, Universidad de Santiago de Compostela and Instituto Galego de F\'isica de Altas Enerxias (IGFAE) E-15782 Santiago de Compostela, Spain}
\author{C. Naya}
\affiliation{Departamento de F\'isica de Part\'iculas, Universidad de Santiago de Compostela and Instituto Galego de F\'isica de Altas Enerxias (IGFAE) E-15782 Santiago de Compostela, Spain}
\author{J. Sanchez-Guillen}
\affiliation{Departamento de F\'isica de Part\'iculas, Universidad de Santiago de Compostela and Instituto Galego de F\'isica de Altas Enerxias (IGFAE) E-15782 Santiago de Compostela, Spain}
\author{J.M. Speight}
\affiliation{School of Mathematics, University of Leeds, Leeds LS2 9JT, England}
\author{A. Wereszczynski}
\affiliation{Institute of Physics,  Jagiellonian University,
Reymonta 4, Krak\'{o}w, Poland}


\begin{abstract}
One problem in the application of the Skyrme model to  nuclear physics is that it predicts too large a value for the compression modulus of nuclear matter.
Here we investigate the thermodynamics of the BPS Skyrme model at zero temperature and calculate its equation of state. Among other results, we find that classically (i.e. without taking into account quantum corrections) the compressibility of BPS skyrmions is, in fact, infinite, corresponding to a zero compression modulus. This 
suggests that the inclusion of the BPS submodel into the Skyrme model lagrangian may significantly reduce this too large value, providing further evidence for the claim that the BPS Skyrme model may play an important role in the description of nuclei and nuclear matter.

\end{abstract}

\maketitle 

\section{Introduction}
The standard Skyrme model \cite{skyrme} has the lagrangian
\begin{equation}
\mathcal{L}= \mathcal{L}_2+\mathcal{L}_4+ \mathcal{L}_0 ,
\end{equation}
where
\begin{equation}
\mathcal{L}_0=-\lambda_0\;  {\cal U}(U), \;\;\; 
\mathcal{L}_2= \lambda_2 \mbox{Tr} \partial_\mu U \partial^\mu U^\dagger, \;\;\; 
\mathcal{L}_4= \lambda_4 \mbox{Tr} ([L_\mu, L_\nu])^2, \;\;\;  
\end{equation}
are the potential, sigma model and Skyrme terms, respectively. Here, $U$ is an SU(2)-valued field of mesons (pions), and $L_\mu = U^\dagger \partial_\mu U$ is the left-invariant current. Further, the potential ${\cal U}$ is a non-negative funtion of $U$ with one unique vacuum. The Skyrme model is considered as an effective theory for low energy QCD. Using results from large $N_c$ expansions, it is known that the proper degrees of freedom in this limit are mesons, while baryons in the Skyrme model emerge as collective excitations, i.e., solitons called skyrmions, with an identification between baryon number and topological charge \cite{thooft}. 
One main problem of the standard Skyrme model is its failure to describe adequately the very small binding energies of physical nuclei. This is related to the fact that, although there exists a topological energy bound in the Skyrme model (\cite{skyrme}, \cite{faddeev}, for improved bounds see \cite{harland}, \cite{gen-bounds}), nontrivial solutions cannot saturate this bound. 
Another problem of the standard Skyrme model, which shall concern us in the present paper, is that it fails to describe simultaneously the hadronic rotational-vibrational excitations (the Roper masses) and the compression modulus of nuclear matter. 

The qualitative reason is the following. The Roper resonances are related to the excitations of the monopole vibrational mode. Technically, they are calculated by quantizing the Derrick scaling factor $\Lambda $ (where the Derrick scaling transformation is $\vec r \to \Lambda \vec r$) and by determining the eigenvalues of the resulting quantum mechanical hamiltonian. In the baryon number $B=1$ sector, the typical fit to nucleon and $\Delta$ resonance masses, although giving quite good values for many observables \cite{anw}, leads to rather unphysical Roper resonances. This situation is improved by coupling the vibrational and rotational modes \cite{roper1}-\cite{roper3}. Then, the Roper masses become bigger but still lighter than the observed resonance energies. On the other hand, assuming that the Derrick rescaling correctly describes the reaction of nuclear matter to the action of external pressure (uniform rescaling), one can easily compute the compression modulus for such a hedgehog solution, using the same model parameters. The resulting value is much higher than the compression modulus of nuclear matter. 

For a more detailed investigation of this problem, let us next introduce a certain generalization of the standard Skyrme model. Indeed,
recently a very special Skyrme type field theory has been proposed \cite{BPS} (our metric conventions are $\eta_{\mu\nu} = {\rm diag}(+,-,-,-)$)
\begin{equation}
\mathcal{L}_{06}=\mathcal{L}_6+ \mathcal{L}_0.
\end{equation}
where 
\begin{equation}
\mathcal{L}_6 = -(24\pi^2)^2 \lambda_6 \mathcal{B}_\mu \mathcal{B}^\mu 
\end{equation}
and 
\begin{equation}
 \mathcal{B}^\mu = \frac{1}{24\pi^2} \epsilon^{\mu \nu \rho \sigma} \mbox{Tr} L_\nu L_\rho L_\sigma, \;\;\; B=\int d^3 x \mathcal{B}^0
\end{equation}
is the baryon (topological) current. The main motivation and advantage of this model, called the BPS Skyrme model, is its BPS property (first noted in \cite{JMP paper}). As a consequence, there are infinitely many solitonic solutions saturating a topological bound, which leads to a linear energy - topological charge relation. Therefore, the classical binding energies are precisely zero. This should be contrasted with the usual Skyrme model, where these energies are too high compared with the experimental values. This is, in fact, the main problem in the application of the Skyrme model to nuclear physics. Non-zero binding energies in the BPS Skyrme model have been derived recently by taking into account the semiclassical quantization of the spin-isospin degrees of freedom, the Coulomb interaction as well as a small isospin breaking. The obtained values are in very good agreement with the nuclear data and the semi-empirical (Weizs\"acker) formula, especially for higher nuclei \cite{BPS bind}, \cite{Marl}. 

Hence, the natural question arises whether the BPS Skyrme model can help to resolve the problem of the predicted compression modulus of nuclear matter being too high, and the predicted Roper masses being too small, as it does with the binding energies. Indeed, it has been shown recently that the Roper masses computed in the BPS Skyrme model are higher than in the standard Skyrme model;  in fact, they are higher than the experimental results \cite{BPS rot-vib}, suggesting that a generalization of the Skyrme model, consisting of all four terms, should give a more realistic description. 

But if one assumes that nuclear matter reacts to external pressure via uniform rescaling, then the Derrick scale parameter $\Lambda$ is the relevant variable for both phenomena, and the sizes both of the compression modulus and of the Roper masses are determined by the same parameter $E^{(2)}\equiv (d^2/d\Lambda^2) E(\Lambda )|$ (the second variation of the energy under rescaling). Here the vertical line $|$ means evaluation at the minimum (at $\Lambda =1$ if only the Derrick scaling is considered; if the scaling (i.e., monopole vibrational) excitation is coupled to other degrees of freedom like, e.g., (iso-)rotational excitations, the minimum may occur at other values). Indeed, in the Roper resonance calculations, $\Lambda$ is quantized ($\Lambda \to \hat \Lambda$) and the harmonic oscillator approximation is used, and $E^{(2)}$ enters directly as the factor multiplying the ``harmonic oscillator potential" $\hat \Lambda^2$. On the other hand, under the assumption of uniform rescaling under external pressure, it may be proved easily that the compression modulus ${\cal K}$ is directly given by $E^{(2)}$ divided by the baryon number (see next section), 
\be \label{stiff-compr}
{\cal K} =(E^{(2)}/B).
\ee
It is a simple exercise to determine $E^{(2)}$ for a generalized Skyrme model which receives contributions from all four terms ${\cal L}_i, i=0,2,4,6$.
We call the corresponding static energies $E_i$ and assume for simplicity that $B=1$ (higher baryon number skyrmions should be approximately $B$ times the results below). Then the total energy is
\be \label{E-tot}
E = E_6 + E_4 + E_2 + E_0 \equiv E_N
\ee
where $E_N$ denotes the energy (mass) of a nucleon. This is not entirely correct, because the nucleon receives further small corrections, e.g., from spin and isospin excitations, but these small corrections are unlikely to be significant for the estimates we are interested in here. For the energy of a rescaled skyrmion solution we get
\be
E(\Lambda ) \equiv E[U(\Lambda \vec r)] = \Lambda^3 E_6 + \Lambda E_4 + \Lambda^{-1} E_2 + \Lambda^{-3} E_0
\ee
and for the second derivative
\be
 E''(\Lambda)= 6\Lambda E_6 + 2 \Lambda^{-3} E_2 + 12 \Lambda^{-5}E_0.
\ee
Using, in addition, the Derrick condition for a solution
\be
 E'(1) = 3E_6 + E_4 - E_2 - 3E_0 \equiv 0
\ee
we get for $E^{(2)}$
\be
E^{(2)} \equiv  E''(1) =  
 E_N + 8(E_6 + E_0) .
\ee
In other words, $E^{(2)}$ is equal to the nucleon mass (the nuclear mass in general) in the model with only ${\cal L}_2$ and ${\cal L}_4$ present (the original Skyrme model), but increases if  ${\cal L}_0$, ${\cal L}_6$ are included. This is good for the Roper resonances, which require a larger value of $E^{(2)}$. On the other hand, if we accept Eq. (\ref{stiff-compr}) then this is obviously bad for the compression modulus. The mass of the nucleon is about $E_N \sim 940 \, {\rm MeV}$, whereas the compression modulus should be about ${\cal K} \sim 230 \, {\rm MeV}$ \cite{modulus}, so contributions from $E_6$ and $E_0$ only make things worse.  

Of course, the arguments above are only qualitative in nature and should be backed up by more detailed computations.
The calculation of the compression modulus for higher nuclei is a very complicated numerical problem which requires the computation of higher skyrmions \cite{skyrmions}, \cite{108}. For the original Skyrme model ${\cal L}_2 + {\cal L}_4$, however, recent results indicate an unacceptably big value of the compression modulus \cite{manton}, essentially confirming the simple arguments from above. Physically, this may be understood as a rather high stiffness of the original Skyrme model, which is possibly related to the crystal structure of Skyrme matter in the limit $B\rightarrow \infty$ \cite{klebanov}.

This crystalline behaviour represents a striking qualitative difference from the BPS Skyrme model. Indeed, the static energy functional of the BPS Skyrme model possesses infinitely many symmetries, among them the volume preserving diffeomorphisms (VPDs) on physical space, which are exactly the symmetries of an ideal liquid. As a consequence, deformations of the classical solitons which do not change their volume cost zero energy. This, of course, does not tell us much about the cost in energy for a deformation which does change the volume, like, e.g., the squeezing of a nucleus (a skyrmion) as a result of external pressure. The qualitative arguments above seem to indicate that this cost in energy is quite high, i.e., the``ideal liquid" provided by the BPS skyrmions is quite incompressible. In the explicit calculations below we shall find that this apparent paradox is resolved by the fact that the uniform rescaling $\vec r \to \Lambda \vec r$ is a very bad approximation for the true behaviour of a classical BPS skyrmion under external pressure. Taking this behaviour correctly into account leads, in fact, to zero compression modulus. This does not mean that it costs zero energy to squeeze a BPS skyrmion under external pressure, it just means that the (infinitesimally small) pressure used to squeeze the nucleon and the resulting small change in volume are not linearly related. 

The rest of the paper is organized as follows. In Section II we describe in detail the zero temperature thermodynamics of the BPS Skyrme model. We explain how to introduce pressure in an analytical way and calculate the volume, the equation of state and the energy of the corresponding skyrmions. We also determine their compressibility and discuss some concrete examples. In Section III we discuss the relevance of our results and explain how they may contribute to resolving the problems with the compression modulus of nuclear matter described by  (generalized) Skyrme models. Finally, in the appendix we prove that our analytical way of introducing the pressure may be generalized to a large class of models, among which the extreme (or BPS) limit of the baby Skyrme model may be found.

\section{$T=0$ Thermodynamics of the BPS Skyrme model}
In thermodynamics, the compressibility at fixed temperature $T$ (isothermal) and particle number $B$ (in our case, particle number equals baryon number $B$) is \cite{landau}
\be 
\kappa_{T,B} = -\frac{1}{V} \left( \frac{\partial V}{\partial P}\right)_{T,B} 
\ee
where $V$ is the volume of the substance and $P$ is the pressure.
This quantity is useful for us because it is defined in terms of global variables (volume $V$ and pressure $P$) which do not vary with position.
Another important quantity is the compression modulus. There exist several definitions \cite{jennings}, which are all equivalent, however, 
for systems with a constant baryon density like, e.g., the free fermi gas (see below). 
Again, we shall use a definition which only depends on global variables, namely
\be \label{comp-mod-T}
{\cal K}= \frac{9V^2}{B} \left( \frac{\partial^2 E}{\partial V^2}\right)_{T,B}.  
\ee
In the sequel, we are only interested in the case of zero temperature, $T=0$. The generalization of skyrmions to non-zero temperature is, in fact, a rather nontrivial problem, see, e.g., \cite{temp} for early attempts.  
At zero temperature, apparently we still have the three thermodynamic variables $P$, $V$ and $B$. For skyrmions, however, the baryon number $B$ is an interger-valued constant which depends only on the boundary conditions imposed on the skyrme field and {\em not} on the thermodynamic state. Specifically, it always holds exactly that $E\propto B$ and $V\propto B$. It is, therefore, more appropriate to treat $B$ as a constant and not as a thermodynamic variable. Further, the volume $V$ and the pressure $P$ are always related by an equation of state $f(P,V)=0$, see below. As a consequence, all thermodynamic functions only depend on {\em one} thermodynamic variable (in the present paper, for convenience we choose the pressure $P$), and the derivatives in the thermodynamic relations are, therefore, ordinary derivatives. The compressibility at zero temperature,  e.g., is
$\kappa = -(1/V) (dV/dP)$, and the compression modulus
\be \label{comp-mod}
{\cal K}= \frac{9V^2}{B} \frac{d^2 E}{d V^2} = \frac{9 V^2}{B} \left( \left( \frac{d V}{d P} \right)^{-2} \frac{d^2 E}{d P^2} -
\left( \frac{d V}{d P}\right)^{-3} \frac{d^2 V}{d P^2} \; \frac{d E}{d P} \right)
\ee
where the second expression is useful if both volume $V(P)$ and energy $E(P)$ are known functions of the pressure $P$, as holds true in our case, see below.  For the free fermi gas, it may be shown easily that $\kappa$ and ${\cal K}$ are related via
\be \label{K-kappa}
{\cal K} = \frac{9V}{B\kappa}.
\ee
With our definition of ${\cal K}$, a sufficient condition for the above relation to hold is just the standard thermodynamical relation $
P = -(\partial F/\partial V)_T$ (where $F$ is the free energy), which at zero temperature reads
\be
P=-\frac{dE}{dV} .
\ee
The above thermodynamic relation is satisfied in the BPS Skyrme model, as we shall see. This last statement is nontrivial, because we do not use a thermodynamic definition of the volume. Our volume is, instead, literally the total space volume occupied by certain topological soliton solutions, see below.

\subsection{Free Fermi gas}

The precise definition of the compression modulus depends on the assumptions made for nuclear matter \cite{jennings}, but the simplest standard definition assumes that nuclear matter is, in a first approximation, just a free Fermi gas of nucleons. In this approximation, the nuclear matter density (baryon density)
\be
\rho = \frac{B}{V}
\ee
is assumed to be spatially constant, which allows one to rewrite the isothermal compressibility as (see, e.g., \cite{preston})
\be
\kappa =   \frac{1}{\rho}\frac{\partial \rho}{\partial P}_{T=0,B}.
\ee
The Fermi momentum for $N $ fermions in a volume $V$, and with degeneracy $D$ (i.e., $D$ fermions may occupy each energy state) is $p_F = \hbar (6\pi^2 N/(DV))^\frac{1}{3}$. In our case $D=4$ (two nucleon species, and two spin degrees of freedom) and $N=B$, so
\be
p_F = \hbar \left( \frac{3}{2}\pi^2 \rho \right)^\frac{1}{3} ,
\ee
and the total kinetic energy due to the exclusion principle is 
\be \label{deg-en}
E_T = \frac{3}{5}BE_F \; , \quad E_F = \frac{p_F^2}{2m_N} = \frac{\hbar^2}{2m_N} \left( \frac{3}{2} \pi^2 \rho \right)^\frac{2}{3} 
\ee
where $m_N$ is the nucleon mass.
Further, the Fermi pressure is
\be
P = \frac{2}{3}\frac{E_T}{V} = \frac{1}{5}\frac{\hbar^2}{m_N} \left( \frac{3}{2} \pi^2 \right)^\frac{2}{3} \rho^\frac{5}{3} ,
\ee
leading to the equation of state 
\be
P = c_{\rm Fg} V^{-\frac{5}{3}} \; , \quad c_{\rm Fg} \equiv \frac{1}{5}\frac{\hbar^2}{m_N} \left( \frac{3}{2} \pi^2 \right)^\frac{2}{3} B^\frac{5}{3}.
\ee
Then, the compression modulus of nuclear matter is defined as
\be
{\cal K} = \frac{1}{p_F^2} \frac{\partial}{\partial p_F} \left( p_F^4 \frac{\partial }{\partial p_F} \frac{E_T}{B} \right) 
= 9 \frac{\partial}{\partial \rho } \left( \rho^2 \frac{\partial}{\partial \rho} \frac{E_T}{B} \right)
\ee
where the second equality follows easily from the definitions above. With these definitions, it may also be shown at once that for the free Fermi gas the compression modulus and the isothermal compressibility are related via (\ref{K-kappa}).
We may also demonstrate easily that, if we use the Skyrme energy (\ref{E-tot}) instead of the degeneracy energy (\ref{deg-en}), then under the same assumption of constant baryon density, the compression modulus is precisely given by (\ref{stiff-compr}). Indeed, we just assume that the density $\rho$ is varied by a scale transformation
\be
\vec r \to \Lambda \vec r \quad \Rightarrow \quad \rho = \Lambda^{-3} \rho_0
\ee
where $\rho_0$ is the constant initial value. With
\be
 d\rho = -3 \rho_0 \Lambda^{-4} d\Lambda
\ee
we get
\be
{\cal K} = \frac{9}{B} \partial_\rho (\rho^2 \partial_\rho E ) = \frac{1}{B} (\Lambda^2 \partial_\Lambda^2 E - 2 \Lambda \partial_\Lambda E)
\ee
and at the equilibrium point $\rho =\rho_0$, i.e., $\Lambda =1$, where $\partial_\Lambda E|_{\Lambda =1}=0$, 
\be
{\cal K} =  (1/B) (\Lambda^2 \partial_\Lambda^2 E)_{\Lambda =1} =(E^{(2)}/B). 
\ee

\subsection{Pressure in the BPS Skyrme model} \label{sec:IIB}

One first difference between the BPS Skyrme model and the simple compressibility calculations of the previous section is given by the fact that the baryon density ${\cal B}_0$ in the BPS Skyrme model is, in general, not constant in space. And we shall see that the different thermodynamical properties may be partially attributed to this difference. Concretely, there exists a certain limit of the BPS Skyrme model (a limiting potential) for which the resulting baryon density {\em is} constant, ${\cal B}_0 =\rho =$ const. Precisely for this limiting case, the simple calculations of the previous section turn out to be completely correct. 

A second, and even more important difference is provided by the BPS nature of the static solutions of the BPS Skyrme model. The relevant property here is the fact that BPS solutions have their pressure identically equal to zero (BPS equations are, for that reason, frequently called ``zero pressure conditions" \cite{bazeia}).  This zero pressure condition allows matter described by the BPS Skyrmion solution to react in a nonlinear way to an infinitesimal external pressure acting on it (the induced change in volume is not proportional to the exerted infinitesimal pressure). 

It is convenient to introduce two new non-negative coupling constants $\mu$ and $\lambda$, so that the static energy density of the BPS Skyrme model is
\be
{\cal E} = \mu^2 {\cal U} + \lambda^2 \pi^4 {\cal B}_0^2
\ee  
and the BPS equation is
\be \label{BPS-eq}
\lambda \pi^2 {\cal B}_0 = \pm \mu \sqrt{\cal U}.
\ee 
On the other hand, the energy-momentum tensor $T^{\mu\nu}$ (and, therefore, the pressure) may be easily computed by introducing a general Lorentzian metric in the lagrangian and by varying the action w.r.t. the metric,
\be
T^{\mu\nu} = -\frac{2}{\sqrt{|g|}} \frac{\delta}{\delta g_{\mu\nu}} \int d^4 x \sqrt{|g|} {\cal L}_{06}
\ee
where $g = {\rm det} \, g_{\mu\nu}$, and the correct expression for the lagrangian for a general metric is
\be
{\cal L}_{06} = -\lambda^2 \pi^4 |g|^{-1} g_{\mu\nu}{\cal B}^\mu {\cal B}^\nu - \mu^2 {\cal U}.
\ee 
For static configurations, where only ${\cal B}_0$ is nonzero, the resulting energy-momentum tensor in Minkowski space is
\bea
T^{00} &=& \lambda^2 \pi^4 {\cal B}_0^2  + \mu^2 {\cal U} \; \equiv \;  {\cal E} \nonumber \\
T^{ij} &=& \delta^{ij} \left(\lambda^2 \pi^4 {\cal B}_0^2 - \mu^2 {\cal U} \right)  \; \equiv \;  \delta^{ij} {\cal P}
\eea
(where ${\cal P}$ is the pressure), and the conservation equations $\partial_\mu T^{\mu\nu} =0$ reduce to
\be \label{const-pr}
{\cal P} = P = {\rm const.},
\ee
so any static solution must have constant pressure. 
First of all, let us remark that this is the energy-momentum tensor of a perfect fluid, which, together with the infinitely many symmetries of the model (the volume-preserving diffeomorphisms) further strengthens the case for its interpretation as a field-theoretic realization of the liquid droplet model of nuclei. 
Next, we observe that the BPS equation (\ref{BPS-eq}) is just the zero pressure condition ${\cal P}=0$.
The constant pressure condition (\ref{const-pr}) is, in fact, completely equivalent to the static field equations. In other words, the static field equations always have a first integral and the pressure is the corresponding integration constant. The only difference between BPS and non-BPS solutions is the (zero or nonzero) value of this integration constant. To prove this, and for later use, we now introduce some notation. Parametrizing the SU(2) Skyrme field like
\be
U= \cos \xi +i \sin \xi \; \vec n \cdot \vec \tau \; , \quad \vec n^2 =1,
\ee
(where $\vec \tau$ are the Pauli matrices), and  
\be
\vec n = (\sin \chi \cos \Phi ,\sin \chi \sin \Phi , \cos \chi ) ,
\ee
the baryon density ${\cal B}_0$ is
\be
{\cal B}_0 = \frac{1}{2\pi^2 }\sin^2 \xi \sin \chi \epsilon^{ijk} \xi_i \chi_j \Phi_k ,
\ee
or, with the notation $\xi^1 = \xi , \xi^2 = \chi, \xi^3 = \Phi$,
\be
{\cal B}_0 = \frac{1}{2\pi^2} M(\xi_a) \epsilon^{ijk}\xi^1_i \xi^2_j \xi^3_k
\ee
where $M(\xi^a)$ is the volume element of the target space $\mathbb{S}^3$. With the help of the algebraic identity
\be
\left( \partial_j \frac{\partial}{\partial \xi^a_j}  - \frac{\partial}{\partial \xi^a} \right) {\cal B}_ 0 = 0
\ee
which may be proved easily \cite{ferreira} (and follows immediately from the fact that ${\cal B}_0$ is a topological density whose Euler-Lagrange variation is identically zero) we find for the Euler-Lagrange variation of the energy density
\be
\left( \partial_j \frac{\partial}{\partial \xi^a_j} - \frac{\partial}{\partial \xi^a} \right) {\cal E} = 2\lambda^2 \pi^4 ( \partial_j {\cal B}_0) \frac{\partial}{\partial \xi^a_j} {\cal B}_0 - \mu^2 \frac{\partial}{\partial \xi^a} {\cal U} \equiv 0 .
\ee
Multiplying the above expression by $\xi^a_k$, summing over $a$ and using the further algebraic identity
\be
\sum_a \xi_k^a \frac{\partial}{\partial \xi^a_j}{\cal B}_0 = \delta_{jk} {\cal B}_0
\ee
leads to the equation
\be
2\lambda^2 \pi^4 (\partial_k {\cal B}_0){\cal B}_0 - \mu^2 \partial_k {\cal U} =0
\ee
which trivially integrates to
\be
\lambda^2 \pi^4 {\cal B}_0^2 - \mu^2 {\cal U} = {\rm const},
\ee
i.e., to the constant pressure condition (\ref{const-pr}), as announced. 

In fact, the observation that fields of constant pressure automatically
satisfy the static field equation holds true for a large class of models
generalizing the BPS Skyrme model. To formulate this precisely, let $(M,g)$,
$(N,h)$ be oriented riemannian $n$-manifolds with volume forms $\vol_M$, $\vol_N$ respectively, $V:N\ra [0,\infty)$ be a smooth potential,
and define the energy of a field $\phi:M\ra N$ to be
\[
E(\phi)=\int_M(\frac12|\phi^*\vol_N|^2+V(\phi))\vol_M.\]
We obtain the BPS Skyrme model by choosing $M=\R^3$, $N=SU(2)=S^3$ with the round 
metric of unit radius, $V=2\mu^2\lambda^{-2}{\cal U}$, and 
$\phi:x\mapsto U(x)$. Note that
 this family also includes the extreme
baby Skyrme model ($n=2$) and the general nonlinear Klein-Gordon model
($n=1$).
As usual, we define the pressure of a field to be (minus) the component of
its 
stress tensor in the direction of $g$, 
which for these models is \cite{jayspesut}
\[
{\cal P}(\phi)=\frac12|\phi^*\vol_M|^2-V(\phi).
\]
In this level of generality, we have the following:

\begin{prop} \label{yuko}
 Let $\phi:M\ra N$ have constant pressure
${\cal P}\geq 0$ in some region $\Omega\subseteq M$. Then $\phi$ satisfies 
the Euler-Lagrange equation for the energy functional $E$ on $\Omega$.
\end{prop}

\noindent The proof, which uses a geometric formulation of the variational
calculus for $E$, is presented in an appendix.

\subsection{Equation of state and compressibility}

The solution of the problem we want to study now consists of solving the static field equation for nonzero pressure and determining the volume of the corresponding solution. First, let us define the class of potentials ${\cal U}$ we want to consider. The potentials depend on $U$ only via ${\rm tr} \, U$, i.e., via $\xi$, and take their unique vacuum value at the north pole, i.e., at $\xi = 0$. Further, we assume that the potentials have a power-like behaviour near the vacuum,
\be \label{near-vac}
\lim_{\xi \to 0} \;  {\cal U}(\xi ) \sim \xi^\alpha \; , \quad \alpha >0 .
\ee 
It is one of the distinguished properties of the BPS Skyrme model that for $0<\alpha <6$ the BPS skyrmion solutions are compactons which differ from their vacuum values only in a bounded region of space. As a consequence, these compacton solutions have finite and well-defined volumes. 
Both the compact BPS skyrmion solutions themselves and their energy and baryon densities continuously join their vacuum values at the boundary. We now want to study  how these compactons change under the influence of external pressure. For nonzero pressure, the energy and baryon number densities at the boundary will no longer be continuous, because by assumption some external forces act on the compactons producing the nonzero pressure. 
In a next step, we assume the spherically symmetric ansatz $\xi = \xi (r)$, $\chi = \chi (\theta)$ and $\Phi = B\phi$ where $(r,\theta, \phi)$ are spherical polar coordinates. Inserting this ansatz into the static field equation (constant pressure equation) and insisting on the correct boundary conditions for skyrmions leads to $\chi = \theta $ and to the equation for $\xi$
\be \label{rad-eq}
\frac{|B|\lambda}{2 r^2} \sin^2 \xi \xi_r = - \mu \sqrt{{\cal U} + \tilde P}
\ee
where $\tilde P = (P/\mu^2 )$, and we chose the minus sign in front of the root because we want to impose the boundary conditions 
\be
\xi (r=0) = \pi \, , \quad \xi (r=\infty )=0 ,
\ee  
leading to a Skyrme field configuration with baryon number $B$. We remark that for our purposes the restriction to the spherically symmetric ansatz is not very restrictive. The reason is that the static energy functional has the base space VPDs as symmetries, so to any spherically symmetric solution there exist infinitely many more solutions with arbitrary shapes and with exactly the same pressure and volume \cite{fosco}, leading to the same thermodynamic relations.  

We introduce the new coordinate
\be
z = \frac{2\mu}{3|B| \lambda}r^3 \quad \Rightarrow \quad r^2 dr = \frac{|B|\lambda}{2\mu} dz
\ee
so that (\ref{rad-eq}) becomes an autonomous ODE,
\be \label{z-xi-ode}
\sin^2 \xi \xi_z = - \sqrt{{\cal U} + \tilde P} .
\ee
Next, we introduce a new coordinate
\be
\eta = \frac{1}{2}\left( \xi - \frac{1}{2} \sin 2\xi \right) \quad \Rightarrow \quad d\eta = \sin^2 \xi d\xi
\ee
and the above equation becomes
\be
\eta_z = - \sqrt{{\cal U} + \tilde P} 
\ee
with boundary conditions
\be
\eta (z=0) = \frac{\pi}{2} \; , \quad \eta (z=\infty)=0 .
\ee
Note that $\eta$ is chosen so that the volume form on $\mathbb{S}^3$ is $\vol_{\mathbb{S}^3} = d\eta \wedge \vol_{\mathbb{S}^2}$.

For specific examples it may be useful to treat the potential ${\cal U}$ as a function of the new coordinate $\eta$, because the resulting ODE is simpler. Near the vacuum, $\eta \sim \xi^3$,  
therefore ${\cal U}(\eta) \sim \eta^\beta$ translates into ${\cal U}(\xi ) \sim \xi^{3\beta}$. 
At the moment, however, we are more interested in generic thermodynamic properties which hold for rather general potentials.
First of all, the volume for general non-negative pressure may be found by integrating  
  Eq. (\ref{z-xi-ode}), which may be re-expressed like
\be
\frac{\sin^2 \xi}{\sqrt{{\cal U} + \tilde P}} d\xi = - dz.
\ee
Integrating both variables over their respective ranges, we get the following integrals
\be
\int_0^\pi \frac{\sin^2 \xi d\xi}{\sqrt{{\cal U} + \tilde P}} = \int_0^Z dz
\ee
or
\be \label{tilde-V}
\tilde V (\tilde P) \equiv  Z(\tilde P)=    \int_0^\pi \frac{\sin^2 \xi d\xi}{\sqrt{{\cal U} + \tilde P}}  = \int_0^{\frac{\pi}{2}} \frac{d\eta}{\sqrt{{\cal U} + \tilde P}}.
\ee
and therefore the volume 
\be \label{V-eq}
V(P) = V(\mu^2 \tilde P) = 2\pi |B|\frac{\lambda}{\mu} \tilde V (\tilde P) .
\ee 
This is the general equation of state of our models. For more explicit expressions, we have to choose specific potentials.
It follows easily from the above expression that for positive pressure $\tilde P >0$, $Z(\tilde P)$ and, therefore, the volume of the Skyrmion,
is finite for arbitrary potentials of the type considered. For BPS Skyrmions ($\tilde P =0$), on the other hand, the volume is finite (the Skyrmion is a compacton) for $0\le \alpha <6$, but infinite for $\alpha \ge 6$.  
For later convenience we also calculate
\be
\frac{d \tilde V}{d \tilde P} = -\frac{1}{2} \int_0^\pi \frac{d\xi \sin^2 \xi}{({\cal U} + \tilde P)^\frac{3}{2}} = 
-\frac{1}{2} \int_0^{\frac{\pi}{2}} \frac{d\eta}{({\cal U} + \tilde P)^\frac{3}{2}} 
\ee
and
\be
\frac{d^2 \tilde V}{d \tilde P^2} = \frac{3}{4} \int_0^\pi \frac{d\xi \sin^2 \xi}{({\cal U} + \tilde P)^\frac{5}{2}} = 
\frac{3}{4} \int_0^{\frac{\pi}{2}} \frac{d\eta}{({\cal U} + \tilde P)^\frac{5}{2}} .
\ee
For the energy we get with the help of the constant pressure equation
\be
E = \int d^3 x \left( \lambda^2 \pi^4 {\cal B}_0^2 + \mu^2 {\cal U} \right) = \int d^3 x \left( 2\mu^2 {\cal U} + P \right) = 4\pi \mu^2 \int dr r^2 \left( 2{\cal U} + \tilde P \right)
\ee
where we used the axially symmetric ansatz in the last step. Using the above variables, we further get
\be
 E(P) = E(\mu^2 \tilde P) = 2\pi \lambda \mu |B| \tilde E(\tilde P)
\ee
where
\be
\tilde E (\tilde P) = \int_0^Z dz \left( 2{\cal U} + \tilde P\right) = \int_0^\pi d\xi \sin^2 \xi \frac{2{\cal U} + \tilde P}{\sqrt{{\cal U} + \tilde P}}
= \int_0^\frac{\pi}{2} d\eta \frac{2{\cal U} + \tilde P}{\sqrt{{\cal U} + \tilde P}}.
\ee
Further,
\be
\frac{d \tilde E}{d \tilde P} = \frac{\tilde P}{2} \int_0^\pi \frac{d\xi \sin^2 \xi}{({\cal U} + \tilde P)^\frac{3}{2}} =
\frac{\tilde P}{2} \int_0^\frac{\pi}{2} \frac{d\eta}{({\cal U} + \tilde P)^\frac{3}{2}}
\ee
\be
\frac{d^2 \tilde E}{d \tilde P^2} = 
\frac{1}{2} \int_0^\pi d\xi \sin^2 \xi \frac{{\cal U} -\frac{1}{2} \tilde P}{({\cal U} + \tilde P)^\frac{5}{2}} =
\frac{1}{2} \int_0^\frac{\pi}{2} d\eta \frac{{\cal U} -\frac{1}{2} \tilde P}{({\cal U} + \tilde P)^\frac{5}{2}} .
\ee
From these results it follows immediately that the thermodynamic relation $P=-(d E/d V)$ holds in the BPS Skyrme model.
Indeed, obviously
\be
\frac{d \tilde E}{d \tilde V} = \frac{\frac{d \tilde E}{d \tilde P}}{\frac{d \tilde V}{d \tilde P}} = -\tilde P
\ee
and, further
\be
\frac{d E}{d V} = \frac{\frac{d E}{d \tilde P}}{\frac{d V}{d \tilde P}} = \mu^2 
\frac{\frac{d \tilde E}{d \tilde P}}{\frac{d \tilde V}{d \tilde P}} = -P.
\ee
Remark: it is interesting to notice that it is precisely the compacton volume which exactly saturates the thermodynamic relation, although it is not directly defined as a thermodynamical quantity. Other measures for the volume which are, e.g., constructed from different charge radii $\langle r \rangle_\gamma $,
\be
\langle r \rangle_\gamma = \left( \int d^3 x r^{3\gamma} {\cal B}_0 \right)^\frac{1}{3\gamma}
\ee
via $V_\gamma =(4\pi/3) \langle r \rangle_\gamma^3$ lead to different results, in general.  In this sense, the compacton volume is singled out as the correct definition of the volume from a thermodynamical perspective.

For the compressibility at equilibrium $\tilde P=0$ we get
\be
\left. \kappa \sim -\frac{1}{Z} \frac{\partial Z}{\partial \tilde P}\right|_{\tilde P =0} = \frac{1}{2Z(0)} \int_0^\pi {\cal U}^{-\frac{3}{2}} 
\sin^2 \xi d\xi .
\ee
Near the vacuum, the integrand on the right hand side behaves like $\xi^{2-\frac{3\alpha}{2}}$. The integral is, therefore, finite for $\alpha <2$ but infinite for $\alpha \ge 2$. It follows that the compression modulus, which is proportional to the inverse of $\kappa$, is zero for $\alpha \ge 2$, as announced.

\subsection{Examples}

At this point we shall treat the potential ${\cal U}$ as a function of the new coordinate $\eta$ for reasons of simplicity, as announced.
Here, instead of calculating the compacton volumes directly from Eq. (\ref{tilde-V}), we shall first calculate the solutions $\eta (z)$ and then determine the compacton volumes from the boundary conditions, because the solutions will be useful on their own.  Remember that near the vacuum, $\eta \sim \xi^3$,  
therefore ${\cal U}(\eta) \sim \eta^\beta$ translates into ${\cal U}(\xi ) \sim \xi^{3\beta}$. For concreteness, we choose the simple potentials
\be
{\cal U} = \eta^\beta
\ee
for different values of $\beta$ as examples. Note that such potentials fail to be differentiable at the anti-vacuum, $\eta = (\pi /2)$. This is unlikely to be significant for our purposes, however. 
\begin{itemize}
\item[i)] $\beta =1$. This corresponds to a cubic potential near the vacuum. The BPS equation (for zero pressure) is
\be
\eta_z = - \eta^\frac{1}{2}
\ee
with the solution ($z_0$ is an integration constant)
\be
\eta = \frac{1}{4} (z_0 - z)^2.
\ee
The condition $\eta(0)=(\pi/2)$ leads to $z_0 = \sqrt{\pi/8}$, and the position $Z$ where $\eta $ reaches its vacuum value, $\eta(z)=0$ for $z\ge Z$ is $Z=z_0 = \sqrt{\pi/8}$. Finally, the volume of the BPS compacton is
\be
V = \frac{4\pi}{3}R^3 = \frac{2\pi |B|\lambda}{\mu} Z =  \frac{2\pi |B|\lambda}{\mu} \sqrt{\frac{\pi}{8}} .
\ee

For nonzero pressure the equation is
\be
\eta_z = -\sqrt{\eta + \tilde P}
\ee
with the solution
\be
\eta = \frac{1}{4} (z_0 - z)^2 - \tilde P
\ee
where $\eta (0)=(\pi/2)$ leads to 
\be 
z_0 = 2 \sqrt{\frac{\pi}{2} + \tilde P},
\ee
whereas $\eta (Z)=0$ leads to
\be \label{tilde-V-beta1}
\tilde V (\tilde P) =Z = z_0 - 2\sqrt{\tilde P} = 2 \left( \sqrt{\frac{\pi}{2} + \tilde P} - \sqrt{\tilde P} \right) .
\ee
The resulting equation of state is
\be
\tilde P = \frac{1}{16 \tilde V^2}\left( 2\pi - \tilde V^2 \right)^2 .
\ee
In agreement with the general discussion of the preceding section, the compressibility is infinite due to
 the presence of the second term proportional to $\sqrt{\tilde P}$ in Eq. (\ref{tilde-V-beta1}),
\be
\kappa = - \frac{1}{V} \frac{d V}{d P}\Bigg|_{P=P_0} \sim \; - \frac{1}{Z} \frac{d Z}{d \tilde P}\Bigg|_{\tilde P =0} = \infty
\ee
(here $P_0$ generically is the equilibrium pressure, which in our case is $P_0=0$). 

Other possible measures for the volume like, e.g., the cubes of average baryon radii of the type
\be
\langle V\rangle_\gamma = \frac{4\pi}{3} \Big| \int d^3 x {\cal B}_0 r^{3\gamma} \Big|^\frac{1}{\gamma}
\ee
which, in our variables, read
\be \label{z-gamma}
\langle V\rangle_\gamma \sim \langle z \rangle_\gamma \equiv 
\Big| \int dz \eta_z z^\gamma \Big|^\frac{1}{\gamma}
\ee
lead to completely different compressibility results. For the simplest case, $\gamma =1$, e.g., we easily calculate
\be
\langle z \rangle_1 =  \Big| \int _0^Z dz \eta_z z \Big| = \frac{2}{3} \left( \left( \frac{\pi}{2} - 2 \tilde P \right)  \sqrt{\frac{\pi}{2} + \tilde P} + 2 \tilde P \sqrt{\tilde P} \right) ,
\ee
which leads to a finite compressibility. These volume definitions, however, should not be used because they do not obey the thermodynamic relation $P=-(d E/d V)$.

\item[ii)] $\beta =(2/3)$. This corresponds to a quadratic (pion mass) potential near the vacuum. The BPS equation (for zero pressure) is
\be
\eta_z = - \eta^\frac{1}{3}
\ee
with the solution ($z_0$ is an integration constant)
\be
\eta = \left( \frac{2}{3} (z_0 - z) \right)^\frac{3}{2}.
\ee
Further, $\eta (0)=(\pi /2)$ leads to
\be
z_0 = \frac{3}{2} \left( \frac{\pi}{2} \right)^\frac{2}{3}
\ee
and the compacton reaches its vacuum value at $Z=z_0$. The equation for nonzero pressure
\be
\eta_z = - \sqrt{\eta^\frac{2}{3} + \tilde P}
\ee
has the implicit solution
\be
\frac{3}{2} \left[ \sqrt{\eta^\frac{2}{3} + \tilde P} \; \eta^\frac{1}{3} - \tilde P \ln \left( 2\left( \sqrt{\eta^\frac{2}{3} + \tilde P} + \eta^\frac{1}{3} \right) \right) \right] = z_0 - z.
\ee
The condition $\eta (0)=(\pi/2)$ leads to
\be
z_0 = \frac{3}{2} \left[ \sqrt{\left( \frac{\pi}{2} \right)^\frac{2}{3} + \tilde P} \; \left( \frac{\pi}{2} \right)^\frac{1}{3} - \tilde P \ln \left( 2\left( \sqrt{\left( \frac{\pi}{2} \right)^\frac{2}{3} + \tilde P} + \left( \frac{\pi}{2} \right)^\frac{1}{3} \right) \right) \right] 
\ee
and $\eta$ reaches its vacuum value $\eta =0$ at
\be
Z = z_0 + \frac{3}{2} (\ln 2) \tilde P + \frac{3}{4} \tilde P \ln \tilde P.
\ee
The volume is, therefore ($\tilde V \equiv Z$),
\be \label{vol-beta23}
\tilde V (\tilde P)  =
\frac{3}{2} \left[ \sqrt{\left( \frac{\pi}{2} \right)^\frac{2}{3} + \tilde P} \; \left( \frac{\pi}{2} \right)^\frac{1}{3} - \tilde P \ln \left( \sqrt{\left( \frac{\pi}{2} \right)^\frac{2}{3} + \tilde P} + \left( \frac{\pi}{2} \right)^\frac{1}{3}  \right) + \frac{1}{2}\tilde P \ln \tilde P\right] .
\ee
Here, the important term is the last one, $\sim \tilde P \ln \tilde P$, because it again leads to an infinite compressibility. This infinitesimal nonlinearity is the softest possible one, therefore we expect finite compressibilities for $\beta < (2/3)$, i.e., $\alpha < 2$, in agreement with the general findings of the previous section.

\item[iii)] $\beta =(1/3)$. This corresponds to a linear (V-shaped) potential near the vacuum. The equation for nonzero pressure is
\be
\eta_z = - \sqrt{\eta^\frac{1}{3} + \tilde P}
\ee
and has the implicit solution
\be
z_0 - z = \frac{2}{5} \sqrt{\eta^\frac{1}{3} + \tilde P} \left( 8 \tilde P^2 - 4 \tilde P \eta^\frac{1}{3} + 3 \eta^\frac{2}{3} \right) 
\ee
which leads to 
\be
z_0 = \frac{2}{5} \sqrt{ \left( \frac{\pi}{2} \right)^\frac{1}{3} + \tilde P} \left( 8 \tilde P^2 - 4 \tilde P \left( \frac{\pi}{2} \right)^\frac{1}{3} + 3 \left( \frac{\pi}{2} \right)^\frac{2}{3} \right) 
\ee
and
\be
Z = z_0 - \frac{2}{5} \tilde P^\frac{5}{2} ,
\ee
i.e., 
\be
\tilde V (\tilde P)= \frac{2}{5} \left( \sqrt{ \left( \frac{\pi}{2} \right)^\frac{1}{3} + \tilde P} \left( 8 \tilde P^2 - 4 \tilde P \left( \frac{\pi}{2} \right)^\frac{1}{3} + 3 \left( \frac{\pi}{2} \right)^\frac{2}{3} \right) - \tilde P^\frac{5}{2} \right) .
\ee
Here, the last term $\sim \tilde P^\frac{5}{2}$ does not contribute to the compressibility at $\tilde P =0$, therefore now the compressibility is finite,
\be
\frac{d Z}{d \tilde P} \Big|_{\tilde P =0} = \frac{d z_0}{d \tilde P}\Big|_{\tilde P =0} = -\sqrt{\frac{\pi}{2}}.
\ee

\item[iv)]
Limit $\beta \to 0$. The potential approaches the Heaviside function, ${\cal U}(\eta) =1$ for $\eta \in [(\pi/2),0)$, and ${\cal U}(\eta =0)=0$. 
The equation for nonzero pressure is
\be
\eta_z = - \sqrt{1+\tilde P} = {\rm const}.
\ee
In this specific case, the baryon density is constant inside the skyrmion. Further, the equation for nonzero pressure may be related to the BPS equation for zero pressure by a simple scale transformation $\vec r \to \vec r' = \Lambda \vec r$,
\be
\eta_{z'} = \Lambda^{-3} \eta_z = - \Lambda^{-3} \sqrt{1+\tilde P} = -1 \; , \quad \Lambda^3 = \sqrt{1 + \tilde P} .
\ee
In other words, the skyrmion for nonzero pressure may be inferred from the skyrmion for zero pressure by a simple, uniform rescaling. 
As a consequence, the simple thermodynamic analysis of the introduction applies to this case. The solution with the correct boundary conditions is 
\be
\eta (z) = \frac{\pi}{2} - \sqrt{1+\tilde P}\, z
\ee
and takes the vacuum value $\eta =0$ at 
\be
Z \equiv \tilde V (\tilde P)=\frac{\pi}{2\sqrt{1+\tilde P}},
\ee
leading to the e.o.s.
\be
\tilde P = \left( \frac{\pi}{2\tilde V}\right)^2 - 1.
\ee
Obviously, the compressibility in this case is finite. 

\item[v)]
At this point it is instructive to consider the case without potential, ${\cal U}=0$. Due to the Derrick theorem, the only acceptable zero pressure solution is the trivial vacuum solution $\eta =0$, but for nonzero pressure the equation
\be
\eta_z = - \sqrt{\tilde P} 
\ee
 has the simple solution
\be \eta = \frac{\pi}{2} - \sqrt{\tilde P} z
\ee
with volume 
\be
Z \equiv \tilde V (\tilde P)= \frac{\pi}{2\sqrt{\tilde P}}
\ee
and equation of state
\be
\tilde P = \left( \frac{\pi}{2\tilde V}\right) ^2 .
\ee
In this case, the compressibility is infinite.

\item[vi)]
$\beta =2$. In this case, the BPS skyrmion is no longer compact, but localized stronger than exponentially (concretely $\sim e^{-z}$). The equation for nonzero pressure is
\begin{equation}
\eta_z=-\sqrt{\eta^2+\tilde{P}}
\end{equation}
with the solution
\begin{equation}
 \eta = \sqrt{\tilde{P}}\sinh (z_0-z) .
\end{equation}
Again we impose the topologically non-trivial boundary conditions
$\eta(z=0)=(\pi /2)$ and  $\eta (z=Z)=0$,
where for nonzero pressure $Z$ turns out to be finite, 
\begin{equation}
Z=z_0, \;\;\; \sinh z_0 = \frac{\pi}{2\sqrt{\tilde{P}}} 
\end{equation}
or 
\begin{equation}
 Z\equiv \tilde V (\tilde P) = \mbox{sinh}^{-1}  \frac{\pi}{2\sqrt{\tilde{P}}} 
= \ln \left( \frac{\pi}{2\sqrt{\tilde P}} + \sqrt{\left( \frac{\pi}{2\sqrt{\tilde P}}\right)^2 + 1} \right) ,
\end{equation}
leading to the equation of state
\be
\tilde P = \left( \frac{\pi}{2\sinh \tilde V}\right)^2 .
\ee
Here, $Z$ is finite for nonzero pressure $P$ but tends to infinity in the limit of zero pressure. We may compute the compressibility as before,
\bea
 - \frac{1}{Z}  \frac{d Z}{d \tilde{P}}\Big|_{\tilde P =0} &=&
- \frac{1}{\mbox{sinh}^{-1} \frac{\pi}{2\sqrt{\tilde{P}}}} \frac{1}{\sqrt{1+ \left( \frac{\pi}{2\sqrt{\tilde{P}}} \right)^2}} \frac{-\pi}{4} \frac{1}{\tilde{P}^{3/2}} \Big|_{\tilde P =0}  \nonumber \\
&=&   \frac{1}{2}  \frac{1}{\mbox{sinh}^{-1} \frac{\pi}{2\sqrt{\tilde{P}}}} \frac{1}{\tilde{P}}\Big|_{\tilde P =0} = \infty .
\eea
As the BPS skyrmion is no longer compact, it might be interesting to calculate one of the cubes of the average baryon radii (\ref{z-gamma}), e.g., 
$\langle V\rangle_1 \sim \langle z \rangle_1$, although the resulting volumes are not the thermodynamic ones.
With (here $t=z_0 - z$)
\bea
\int_0^{z_0} \eta_z z dz  &=& \sqrt{\tilde{P}} \int_{z_0}^0 (z_0-t) \cosh t dt =
 \sqrt{\tilde{P}} \left( - z_0 \sinh z_0 +z_0 \sinh z_0 -\cosh z_0 +1 \right) \nonumber \\ &=&
 \sqrt{\tilde{P}} \left( -\cosh z_0 +1 \right) 
\eea
and $\int dz \eta_z =-(\pi /2) = -\sqrt{\tilde P} \sinh z_0$
we get
\begin{equation}
\langle z \rangle_1 =\frac{\cosh z_0 -1 }{\sinh z_0} =  \sqrt{1+\frac{1}{\sinh^2 z_0}} -  \frac{1}{\sinh z_0}
=\sqrt{1+\frac{4\tilde{P}}{\pi^2}} -  \frac{2\sqrt{\tilde{P}}}{\pi}
\end{equation}
and 
\begin{equation}
\left. -\frac{1}{\langle z\rangle_1 } \frac{d \langle z\rangle_1}{d \tilde{P}} \right|_{\tilde{P}=0} = \infty .
\end{equation}
The functional dependence of $\tilde V$ and $\langle z\rangle_1$ on $\tilde P$ is, again, completely different, $\langle z\rangle_1$ being very similar to the thermodynamic volume (compacton volume) for $\beta =1$ in this case.

\end{itemize}

\section{Discussion}

It was one main purpose of the present paper to calculate skyrmion solutions for nonzero external pressure and to determine the resulting thermodynamic properties, concretely the equation of state, the energy and the (isothermal) compressibility. In our specific calculations, we restricted to the BPS Skyrme model, because due to its integrability and BPS properties, all calculations can be done essentially in an analytic fashion. Indeed, all static solutions of the BPS Skyrme model have constant pressure, and the general static field equations are equivalent to (may be once integrated to) the constant pressure condition, where the pressure is the integration constant. Further, the BPS solutions correspond to stable zero pressure solutions, whereas solutions with non-zero pressure require the action of external pressure to be stabilized. 

First of all, let us remark that the same results continue to hold for a large class of generalizations of the BPS Skyrme model, as proved in section \ref{sec:IIB}. The thermodynamics of these models will, therefore, be similar and allow for an equivalent treatment. Among these models is the extreme (or BPS) limit of the baby Skyrme model, whose thermodynamic properties may be of direct physical relevance, because the baby Skyrme model has some applications to condensed matter systems.

In our explicit examples for the BPS Skyrme model, we used the potentials ${\cal U}(\eta) = \eta^\beta$ [where $\eta$ is related to ${\rm tr} \; U = 2\cos \xi$ via $\eta = (1/2)(\xi -\sin \xi \cos \xi )$] for some specific values of the parameter $\beta$, because of the resulting simple exact solutions, even for nonzero pressure. We were able to demonstrate, however, that both the (compact or non-compact) nature of the solutions and the resulting compressibilities are determined exclusively by the behaviour of the potentials near the vacuum. 

For compact solutions (compactons) we further found that, among
all possible volume definitions for a skyrmion, the compacton volume is singled out as especially ``physical" or ``natural" because it saturates the thermodynamic relation $P = -(d E/d V)$, although none of these three quantities is defined to obey this relation, at least not in an obvious way. More concretely, the energy and the pressure are related rather closely via the energy-momentum conservation, but there is no obvious close relation with the compacton volume. The deeper reason behind this fact is probably the BPS property of the BPS Skyrme model, although this should be investigated further.

For different potentials of the above family (i.e., for different asymptotic behaviour), we find the following compressibility results in the BPS Skyrme model. 
\begin{itemize}
\item[i)]
The constant pressure condition (\ref{const-pr}) leads to ${\cal B}_0 \sim \pm \sqrt{{\cal U} + \tilde P}$, so for non-constant potentials the baryon density cannot be constant, either. There exists, however, a limiting case $\beta \to 0$ with a potential which is constant and jumps to zero at the vacuum value $\eta =0$, leading to a constant baryon density which jumps to zero at the compacton boundary. In this limiting case, the skyrmion responds to external pressure via a simple uniform rescaling, and the standard thermodynamic arguments apply.
\item[ii)] For $0<\beta <(2/3)$, i.e., for an asymptotic behaviour ${\cal U} \sim \xi^\alpha$ with $0<\alpha <2$, the compressibility is still finite, but bigger than in the $\beta =0$ case (i.e., the skyrmions are more compressible than in the constant density limit).
\item[iii)] For $(2/3) \le \beta <2$ (i.e., for asymptotic behaviour $2\le \alpha <6$), the isothermal compressibility is infinite, corresponding to a  
zero compression modulus. On the other hand, for these parameter values the BPS skyrmions are still compactons with a well-defined volume.

\end{itemize}  
So we found that for potentials with an asymptotic behaviour about the vacuum which is at least quadratic, i.e., $\alpha \ge 2$, BPS skyrmions have infinite compressibility. Further, potentials with $\alpha <2$ are problematic (their second variation about the vacuum is infinite), therefore $\alpha \ge 2$ are the physically acceptable values, and the infinite compressibility at the equilibrium point $P=0$ is a rather generic result.

At this point, we want to add the following observation. Firstly, we did not worry much about the precise value of the baryon number $B$, because both the volumes and the energies are exactly linear in $B$, whereas the pressure does not depend on it. Secondly, we only considered the case of non-negative pressure $P\ge 0$ or, equivalently, the case where the thermodynamic or compacton volume is less than or equal to its equilibrium value $V_0 \equiv V (P=0)$. If we want to go beyond this case, then the difference between $B=1$ and large $B$ becomes essential. For $B=1$, no physically sensible solution for negative $P$ can be given. It follows immediately from the constant pressure equation for negative $P$ that the skyrmion can never reach its vacuum value where ${\cal U}(U)=0$, because this would lead to an imaginary baryon density. Solutions where the baryon density goes to zero for large radii may exist but still lead to infinite energy,
\be
E = \int d^3 x \left( \lambda^2 \pi^4 {\cal B}_0^2 + \mu^2 {\cal U}  \right) = \int d^3 x \left( 2 \lambda^2 \pi^4 {\cal B}_0^2 + |P| \right)
\ge \int d^3 x |P| = \infty .
\ee
For large baryon number $B$, however, and for potentials with $\alpha <6$, there exists a different possibility for states with $V>V_0$. For such potentials the equilibrium solutions for $P=0$ are compactons, therefore states consisting of a collection of non-overlapping compactons may be formed such that the additional available volume $\delta V = V-V_0$ is occupied by the empty space (vacuum) surrounding these non-overlapping compactons. The pressure of these configurations is obviously zero. In other words, for large $B$ the equilibrium volume $V_0$ defines a phase transition. For $V>V_0$, the system is in the state of an ideal gas of non-overlapping compactons at zero temperature, with zero pressure. For $V<V_0$, on the other hand, the system is in a kind of liquid phase with a rather nontrivial equation of state even at zero temperature.  In this picture, the equilibrium (compacton) volume $V_0$ corresponds to the condensation volume where all empty space surrounding the gas of "molecules" (compactons) has been expelled and the condensation to a liquid phase sets in. For illustrative purposes we plot the corresponding equation of state for the potential ${\cal U} = \eta^\frac{2}{3}$ (i.e., for the case of a pion-mass type potential with $\alpha =2$) in Fig. 1, showing both the liquid (for $0\le V \le  (3/2)(\pi /2)^\frac{2}{3} \simeq 2.027$) and the gaseous phase (for $V>  (3/2)(\pi /2)^\frac{2}{3} $). Here, the dimensionless expression (\ref{vol-beta23}) is used for the plot. We remark that qualitatively similar equation of state diagrams, specifically with the same phase transition, are also found in more conventional calculations of the nuclear equation of state at zero temperature, based on microscopic two- and three-body internuclear forces, see, e.g.,  Fig. 10 of Ref. \cite{well1}.  

\begin{figure}[h]
\begin{center}
\includegraphics[width=0.5\textwidth]{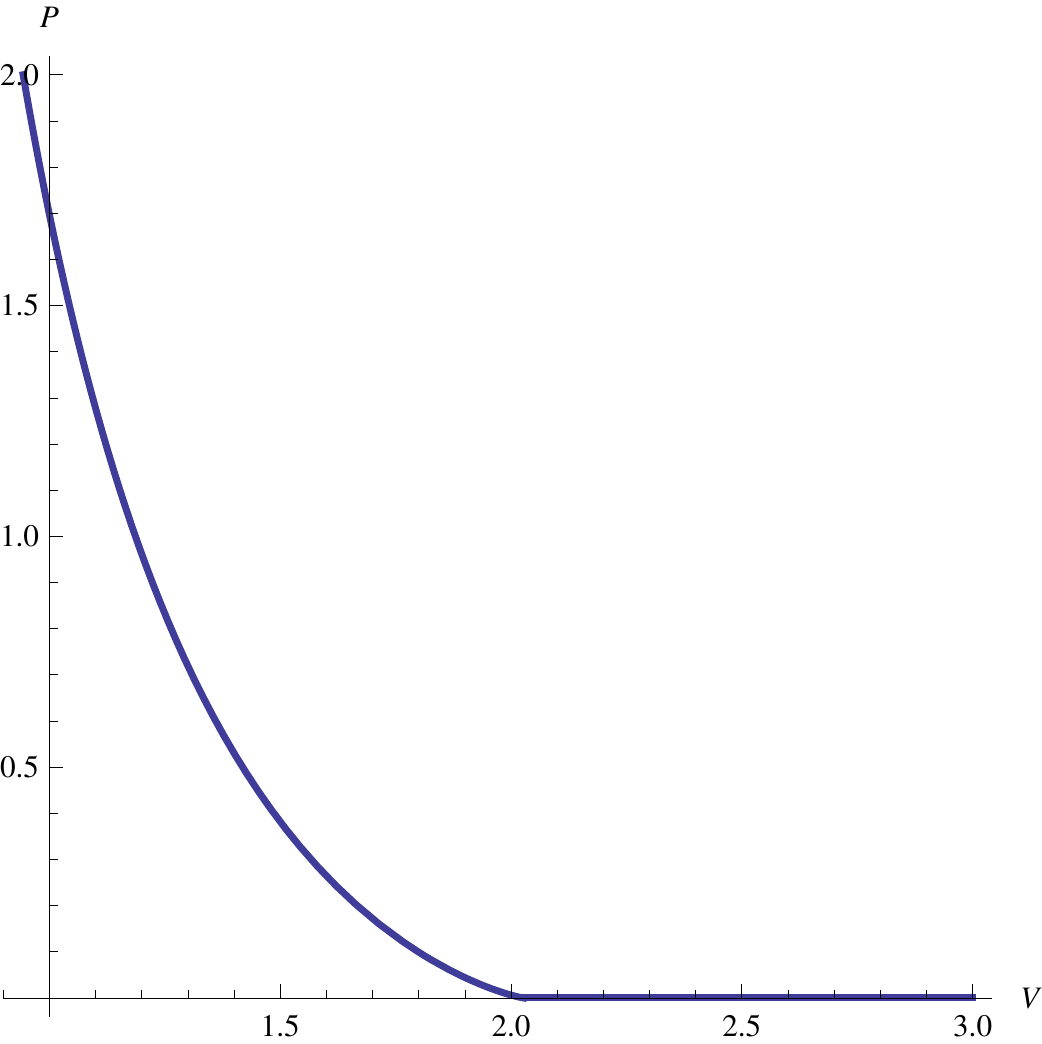}
\caption{Equation of state for the potential ${\cal U} = \eta^\frac{2}{3}$. }
\label{Fig-eos}
\end{center}
\end{figure}

The obvious question now is what our results imply for the problem of the too high compression modulus for nuclear matter described by Skyrme models. As explained in the Introduction, one possible underlying source of the problem is that if a simple uniform (Derrick) rescaling of the baryon density under external pressure is assumed, then it may be shown easily that skyrmions are much more incompressible than physical nuclear matter. Further, it may be shown with the help of standard thermodynamic arguments that for nuclear matter with a constant baryon density the assumption of uniform rescaling under external pressure does apply. If we restrict to the BPS Skyrme model, then the paradox is resolved by the observation that the baryon density for BPS skyrmions is not constant and, consequently, a BPS skyrmion does not respond with a uniform rescaling to external pressure. Our explicit results for the thermodynamcs of classical BPS skyrmions show that the true compression modulus in this case is, in fact, zero. In other words, BPS skyrmions react rather differently to an excitation of the simplest degree of freedom (the uniform rescaling) on the one hand, and to an adiabatic compression to a new true equilibrium state (a static solution) of constant pressure, on the other hand, being quite hard (incompressible) in the former case, but much softer in the latter.

The BPS Skyrme model, however, is only an approximation, whereas a more complete description of nuclei within the Skyrme model framework certainly requires the inclusion of further terms. 
The determination of the volumes, energies and compressibilities for more general Skyrme models then requires the solution of the corresponding Euler-Lagrange equations for nonzero external pressure. There are no longer infinitely many symmetries at our disposal which would allow to change the shapes of solutions, therefore these solutions will have definite shapes. For higher baryon number, like their zero pressure counterparts (standard skyrmions), they will, in general, not be spherically or axially symmetric but preserve some discrete symmetries, at most. Their determination is, therefore, a complicated numerical problem. One may try, however, to achieve a more modest goal, namely the numerical determination of the $B=1$ hedgehog for nonzero external pressure and, hence, its thermodynamic properties. The compressibility of the hedgehog will not be exactly equal to the compressibility of nuclear matter, which corresponds to the case of large $B$, but it might, nevertheless, provide us with some approximate or qualitative information. For the hedgehog we have the spherical symmetry at our disposal (i.e., the skyrmion profile $\xi$ only depends on the radius $r$), therefore the condition of nonzero pressure may be implemented simply as a boundary condition. Indeed, a skyrmion with constant pressure $P$ and radius $r=R$ is given by a profile $\xi (r)$ which obeys
\be
\xi (R)=0 \quad \mbox{and} \quad {\cal P}(r=R) \equiv \frac{1}{3} \sum_i T_{ii} (r=R) =P={\rm const.}
\ee
We remark that the energy-momentum tensor for generalized Skyrme models is more complicated. Specifically, it is no longer true that the pressure is constant in the interior of the skyrmion, so the constant pressure condition has to be implemented as a boundary condition at the skyrmion surface. Some first numerical results for the model $ {\cal L}_2 + {\cal L}_0 + {\cal L}_6$ (i.e., the BPS Skyrme model plus the standard nonlinear sigma model term) indicate that the resulting compressibility $\kappa = -V^{-1} (\partial V/\partial P)_{P=0}$ is still infinite. 

One first possible generalization is to directly use the definition (\ref{comp-mod}) of the compression modulus for the same generalized Skyrme model $ {\cal L}_2 + {\cal L}_0 + {\cal L}_6$. In this model, again, the volume is not defined thermodynamically, so the resulting thermodynamics might be more complicated, with no direct relation between the compressibility (which apparently still is infinite) and the compression modulus (which might then be non-zero). 

Another direction for further investigation is motivated by the following observation. The physical measurements which give rise to the original problem, i.e., the measurements of the Roper resonance and the compression modulus of nuclear matter, are, in fact, measurements of quantum excitations in both cases, namely of the proper Roper resonance and of the excitation energies of giant monopole resonances, respectively. They may both be related to the same simple classical quantity (the compression modulus defined like ${\cal K} = E^{(2)}/B$, see the Introduction) precisely because of the assumption of a simple uniform rescaling. Indeed, if uniform rescaling is assumed, then the Derrick scaling parameter $\Lambda$ appears in the resulting effective action as a variable (a collective coordinate) whose quantization directly leads to a harmonic oscillator. The classical compression modulus shows up in this quantum harmonic oscillator as a parameter multiplying the harmonic oscillator potential. 
We know, however, that at least BPS skyrmion matter does not respond with uniform rescaling to an external force or pressure. In other words, the Derrick parameter $\Lambda$ is not the softest monopole mode (i.e. the softest excitation which respects the rotational symmetry). 
  The proposal, therefore, is to quantize the pressure $P$. By this we mean the following. We may interpret $P$ just as a parameter which describes a possible spherically symmetric deformation of the original skyrmion. In other words, $P$ is a collective coordinate (not a zero mode, but a parameter which describes a collective degree of freedom). 
As the resulting deformed skyrmion still obeys the static field equations, this is, in fact, the softest possible deformation which goes from the old to the new boundary conditions (from the old to the new compacton volume). It should, therefore, correspond to the softest possible monopole vibrational mode, whose excitation energies may be calculated by quantizing this collective coordinate. The true compression modulus of (BPS) skyrmion matter should then be extracted from these excitation energies. 

These issues are under current investigation.

To summarize, we think that our results on the thermodynamics of BPS skyrmions will be instrumental in the resolution of the puzzle of the high compression modulus (too high stiffness) of the Skyrme model. More generally, these results should provide a first step towards the goal of a reliable description of nuclear thermodynamics within the framework of (generalized) Skyrme models (for a recent study of nuclear thermodynamics we refer to \cite{well1}, where also further references can be found). As said, both the inclusion of further terms into the lagrangian and numerical methods will be required for a more quantitative and more precise study of nuclear matter and its equation of state. Another interesting further step consists in the inclusion of the gravitational interaction into the model, which should then allow to study the formation of neutron stars and their equation of state within the BPS Skyrme model and its generalizations.    

\section*{Appendix: proof of Proposition \ref{yuko}}

Let $(M,g,\vol_M)$, $(N,h,\vol_N)$ be oriented riemannian $n$-manifolds, 
$V:N\ra [0,\infty)$ be smooth, and the energy of a field $\phi:M\ra N$
be 
\[
E(\phi)=\frac12\|\phi^*\vol_N\|_{L^2}^2+\int_MV\circ\phi,
\]
as in section \ref{sec:IIB}. Given a vector bundle ${\sf E}$ over $N$ we
denote 
by $\phi^{-1}{\sf E}$ its pullback to $M$. Associate to any field $\phi$
the section $\mu(\phi)\in\Gamma(\phi^{-1}T^*N)$ which maps $A\in
T_{\phi(x)}N$
to 
\[
\mu(\phi)(A)=\ip{\delta\phi^*\vol_N,\iota_A\phi^*\vol_N}_g
\]
where $\delta=-*\d*$ is the $L^2$ adjoint of the exterior differential $\d$,
and $\iota$ denotes interior product. Let the {\em tension field} of $\phi$
be
\[
\tau(\phi)=-\sharp_h\mu(\phi)-(\grad V)\circ \phi\in\Gamma(\phi^{-1}TN)
\]
where $\sharp_h:T^*N\ra TN$ denotes the metric isomorphism induced by $h$.

\begin{lemma} Let $\phi_t$ be a smooth variation of $\phi=\phi_0:M\ra N$,
with infinitesimal generator
$X=\partial_t\phi_t|_{t=0}\in\Gamma(\phi^{-1}TN)$.
Then
$$
\frac{d\: }{dt}\bigg|_{t=0}E(\phi_t)=-\ip{X,\tau(\phi)}_{L^2},
$$
that is, the Euler-Lagrange equation for $E$ is precisely $\tau(\phi)=0$.
\end{lemma}

\begin{proof}
By the homotopy lemma,
$\partial_t|_{t=0}\phi_t^*\vol_N=\d(\iota_X\phi^*\vol_N)$, so
\begin{eqnarray*}
\frac{d\:
}{dt}\bigg|_{t=0}E(\phi_t)&=&\ip{\phi^*\vol_N,d(\phi^*\iota_X\vol_N}_{L^2}+\int_MdV_\phi(X)\\
&=&\ip{\delta(\phi^*\vol_N),\phi^*\iota_X\vol_N}_{L^2}+\ip{(\grad V)\circ
\phi,
X}_{L^2}\\
&=&-\ip{\tau(\phi),X}_{L^2}.
\end{eqnarray*}
\end{proof}

Let $F_\phi:M\ra\R$ denote the function $*\phi^*\vol_N$, so $\phi^*\vol_N
=F_\phi\vol_M$, and ${\cal P}=\frac12F_\phi^2-V\circ\phi$. We wish to
prove that ${\cal P}\geq 0$ being constant implies $\tau(\phi)=0$. This
follows 
quickly from the following lemma, originally proved in the special case
 $n=2$ in
in \cite{speight-semicompactons}:

\begin{lemma} For any smooth map $\phi:M\ra N$ and vector field $Y$ on $M$,
$$
h(\tau(\phi),d\phi(Y))=(\d{\cal P})(Y).
$$
\end{lemma}

\begin{proof} We note that $\delta(\phi^*\vol_N)=-*\d F_\phi$ and, for
$X=\d\phi(Y)$, $\phi^*(\iota_X\vol_N)=\iota_Y\phi^*\vol_N$. Hence
\begin{eqnarray*}
h(\tau(\phi),d\phi(Y))&=&-\ip{-*\d F_\phi,\iota_Y(F_\phi\vol_M)}_g
-\ip{(\grad V)\circ \phi,\d\phi(Y)}_h\\
&=&F_\phi\ip{**\d F_\phi,*\iota_Y\vol_M}_g
-\ip{(\grad V)\circ \phi,\d\phi(Y)}_h\\
&=&F_\phi\ip{(-1)^{n+1}\d F_\phi,(-1)^{n+1}\flat_g Y}_g
-\ip{(\grad V)\circ \phi,\d\phi(Y)}_h\\
&=&\frac12 \d(F_\phi^2)(Y)-\d(V\circ\phi)(Y),
\end{eqnarray*}
where $\flat_g:TM\ra T^*M$ is the metric isomorphism induced by $g$.
\end{proof}

We note in passing that it follows immediately from this lemma that
static solutions have constant pressure ($\tau(\phi)=0$ implies $\d{\cal
P}=0$).

Conversely, let $\phi$ have constant pressure ${\cal P}\geq 0$ on some
region $\Omega\subseteq M$. If ${\cal P}=0$, then $\phi$ is BPS, and hence
automatically solves the static field equation, so assume ${\cal P}>0$. 
Then $\phi$ has no critical points in $\Omega$
(if ${\rm rank}(\d\phi_x)<n$ then $F_\phi(x)
=0$, so ${\cal P}=-V(\phi(x))\leq 0$, a contradiction). Since $\d{\cal
P}_x=0$, 
$\tau(\phi)(x)\in T_{\phi(x)}N$ is orthogonal to $\d\phi_x(T_xM)$ for each
$x\in\Omega$. But $\d\phi_x(T_xM)=T_{\phi(x)}N$ since $\d\phi_x$ has maximal
rank. Hence $\tau(\phi)(x)=0$.

\section*{Acknowledgement}
The authors acknowledge financial support from the Ministry of Education, Culture and Sports, Spain (grant FPA2008-01177), 
the Xunta de Galicia (grant INCITE09.296.035PR and
Conselleria de Educacion), the
Spanish Consolider-Ingenio 2010 Programme CPAN (CSD2007-00042), FEDER and the UK Engineering and Physical Sciences Research Council. 
CN thanks the Spanish
Ministery of
Education, Culture and Sports for financial support (grant FPU AP2010-5772).
Further, AW was supported by polish NCN (National Science Center) grant DEC-2011/01/B/ST2/00464 (2011-2014). AW thanks Nick Manton for stimulating discussions on the issue of incompressibility in the Skyrme models. Further, he thanks S. Krusch for useful discussions. CA, CN and JSG thank G. Torrieri and R. Vazquez for helpful discussions.

\end{document}